\documentclass[preprint]{imsart}
\usepackage{amsmath,amssymb,amsfonts,graphicx,amsthm,nicefrac,mathtools,bm}
\usepackage[numbers]{natbib}
\usepackage[english]{babel}
\usepackage{times}
\usepackage[T1]{fontenc}
\usepackage{bbm}
\usepackage{color}
\usepackage{xspace}
\usepackage{rotating,multirow}

\startlocaldefs
\newtheorem{theorem}{Theorem}

\newtheorem{lemma}{Lemma}

\theoremstyle{definition}
\newtheorem{definition}{Definition}
\theoremstyle{remark}

\makeatletter
\newtheorem*{rep@theorem}{\rep@title}
\newcommand{\newreptheorem}[2]{%
\newenvironment{rep#1}[1]{%
 \def\rep@title{#2 \ref{##1}}%
 \begin{rep@theorem}}%
 {\end{rep@theorem}}}
\makeatother
\newreptheorem{theorem}{Theorem}
\newreptheorem{lemma}{Lemma}
\newreptheorem{corollary}{Corollary}
\newreptheorem{proposition}{Proposition}

\newcommand{\figref}[1]{Figure~\ref{fig:#1}}

\newcommand{\secref}[1]{Section~\ref{sec:#1}}

\newcommand{\thmref}[1]{Theorem~\ref{thm:#1}}

\newcommand{\eqnref}[1]{\eqref{eqn:#1}}

\DeclareMathOperator{\diag}{diag}

\DeclareMathOperator{\sign}{sign}

\DeclareMathOperator*{\argmin}{arg\,min}

\newcommand{\iid}[0]{i.i.d.\xspace}

\newcommand{\norm}[1]{\lVert{#1}\rVert}

\newcommand{\EE}[1]{\mathbb{E}\left[{#1}\right]} 

\def\R{\mathbb{R}}

\newcommand{\ident}{\mathbf{I}}

\newcommand{\iidsim}{\stackrel{\mathrm{iid}}{\sim}}

\newcommand{\ignore}[1]{}

\newcommand{\eps}{\epsilon}

\newcommand{\bbSig}{\Sigma\hspace{-3.5pt}{\color{white}1}\hspace{-7pt}\Sigma}

\newcommand{\normal}{\mathcal{N}}
\newcommand{\fronorm}[1]{\norm{{#1}}_{\text{Fro}}}
\newcommand{\grpnorm}[1]{\norm{{#1}}_{\text{group}}}
\newcommand{\Xp}{\widetilde{X}}
\newcommand{\bbX}{\mathbb{X}}
\newcommand{\bbXp}{\widetilde{\bbX}}
\newcommand{\fdr}{\textnormal{FDR}}
\newcommand{\mfdr}{\textnormal{mFDR}}
\newcommand{\gfdr}{\fdr_{\textnormal{group}}}
\newcommand{\mgfdr}{\mfdr_{\textnormal{group}}}
\newcommand{\fdp}{\textnormal{FDP}}
\newcommand{\fdph}{\widehat{\fdp}}
\newcommand{\Sh}{\widehat{S}}
\newcommand{\St}{\widetilde{S}}
\newcommand{\betah}{\widehat{\beta}}
\newcommand{\tlam}{\widetilde{\lambda}}
\newcommand{\hlam}{\widehat{\lambda}}

\endlocaldefs

\begin{document}

\begin{frontmatter}

\title{The knockoff filter for FDR control  in group-sparse and multitask regression}
\runtitle{The knockoff filter for FDR control  in group-sparse and multitask regression}


\author{\fnms{Ran} \snm{Dai}\ead[label=e1]{randai@uchicago.edu}}
\address{\printead{e1}}
\and
\author{\fnms{Rina Foygel} \snm{Barber}\ead[label=e2]{rina@uchicago.edu}}
\address{\printead{e2}}
\affiliation{Department of Statistics, The University of Chicago, Chicago IL 60637 USA}

\runauthor{R. Dai and R.F. Barber}

\begin{abstract}
We propose the group knockoff filter,
 a method for false discovery rate control in a linear regression setting
where the features are grouped, and we would like to select a set of relevant groups
which have a nonzero effect on the response. By considering the set of true 
and false discoveries at the group level, this method gains power relative to 
sparse regression methods. We also apply our method to the multitask regression
problem where multiple response variables share similar sparsity patterns
across the set of possible features. Empirically, the group knockoff filter successfully
controls false discoveries at the group level in both settings, with substantially
more discoveries made by leveraging the group structure.
\end{abstract}



\end{frontmatter}

\section{Introduction}
In a high-dimensional regression setting, we are faced with many potential explanatory variables (features),
often with most of these features having zero or little true effect on the response. Model
selection methods can be applied to find a small submodel containing the most relevant features,
for instance, via
sparse model fitting methods such as the lasso \cite{tibshirani1996regression}, 
or in a setting where the sparsity respects a grouping of the features, the group lasso
\cite{yuan2006model}. In practice, however, we may not be able to determine whether 
the set of features (or set of groups of features) selected might contain many false positives.
For the (non-grouped) sparse setting, the knockoff filter \cite{barber2015} creates ``knockoff copies''
of each variable to act as a control group, detecting whether the lasso (or another 
model selection method) is successfully controlling the false discovery rate (FDR), and tuning
this method to find a model as large as possible while bounding FDR. In this work,
we will extend the knockoff filter to the group sparse setting, and will find that
by considering features, and constructing knockoff copies, at the group-wise level,
we are able to improve the power of this method at detecting true signals. Our method
can also extend to the multitask regression setting \cite{obozinski2006multi}, where multiple
responses exhibit a shared sparsity pattern when regressed on a common set of features.
As for the knockoff method, our work applies to the setting where $n\geq p$.

\section{Background}
We begin by giving background on several models and methods underlying our work.
\subsection{Group sparse linear regression}
We consider a linear regression model,
$Y=X\beta+z$,
where $y\in \R^{n}$ is a vector of responses and $X\in\R^{n\times p}$ is a known design matrix. 
In a grouped setting, the $p$ features are partitioned into $m$ groups of variables, $G_1,\dots,G_m\subseteq \{1,\dots,p\}$, with group sizes $p_1,\cdots,p_m$. The noise distribution is assumed to be $z\sim\normal(0,\sigma^2\ident_n)$. We assume sparsity structure in that only a small portion of $\beta_{G_i}$'s are nonzero, where $\beta_{G_i}\in\R^{p_i}$ is the subvector of $\beta$ corresponding to the $i$th group of features. When not taking group into consideration, a commonly used method to find a sparse vector of coefficients $\beta$ is the lasso \cite{tibshirani1996regression}, an $\ell_1$-penalized linear regression, which minimizes  the following objective: function
 \begin{equation}\label{eqn:lasso}\betah (\lambda) = \argmin_{\beta} \left\{\norm{y - X \beta}_2^2 +\lambda \norm{\beta}_1\right\}\;.\end{equation}
To utilize the feature grouping, so that an entire group of features is selected simultaneously, \citet{yuan2006model} proposed following grouped lasso penalties:
 \begin{equation}\label{eqn:glasso}\betah (\lambda) = \argmin_{\beta} \left\{\norm{y - X \beta}_2^2 +\lambda\grpnorm{\beta}\right\}\;.\end{equation}
 where $\grpnorm{\beta}=\sum_{i=1}^m \norm{\beta_{G_i}}_2$.
This penalty promotes sparsity at the group level; for large $\lambda$, few groups will be selected (i.e.~$\beta_{G_i}$ will be zero for many groups),
but within any selected group, the coefficients will be dense (all nonzero). The $\ell_2$ norm penalty on $\beta_{G_i}$ may sometimes be rescaled relative to the size of the group.

\subsection{Multitask learning}
In a multitask learning problem with a linear regression model, we consider the model
\begin{equation}\label{eqn:multitask_model}Y = X B + E\end{equation}
where the response $Y\in\R^{n\times r}$ contains $r$ many response variables measured for $n$ individuals,
$X\in\R^{n\times p}$ is the design matrix,  $B\in\R^{p\times r}$ is the coefficient matrix, and $E\in\R^{n\times r}$ is the error
matrix, for which we assume a Gaussian model: its rows $e_i$, 
for $i=1,\dots,n$, are \iid draws from a zero-mean Gaussian,
$e_n \iidsim \normal(0,\Sigma)$,
with unknown covariance structure $\Sigma\in\R^{r\times r}$. If the  number of features $p$ is large, we may believe that
only a few of the features are relevant; in that case, most rows of $B$ will be zero---that is, $B$ is row-sparse.

In a low-dimensional setting, we may consider the multivariate normal model, with likelihood
determined by both the coefficient matrix $B$ and the covariance matrix $\Sigma$.
in a high-dimensional setting, combining this likelihood with a sparsity-promoting penalty may be computationally challenging, and so a common approach
is to ignore the covariance structure of the noise and to simply use a least-squares loss together with a penalty,
\begin{equation}\label{eqn:multitask_lasso}\widehat{B} = \argmin_B \left\{\frac{1}{2}\fronorm{Y - X B}^2 +\lambda \norm{B}_{\ell_1/\ell_2}\right\}\;,\end{equation}
where $\fronorm{\cdot}$ is the Frobenius norm,
and where the $\ell_1/\ell_2$ norm in the penalty is given by
$\norm{B}_{\ell_1/\ell_2} = \sum_i \sqrt{\sum_j B_{ij}^2}$.
This penalty promotes row-wise sparsity of $B$: for large $\lambda$, $\widehat{B}$ will have many zero rows, however the nonzero rows will
themselves be dense (no entry-wise sparsity).

It is common to reformulate this $\ell_1$-penalized multitask linear regression as a group lasso problem. First, we reorganize the terms in our model. We
form a vector response $y\in\R^{nr}$ by stacking the columns of $Y$:
 \[y=\textnormal{vec}(Y)
  = (Y_{11}, \dots, Y_{n1}, \dots, Y_{1r},\dots, Y_{nr})^\top\in\R^{nr},\]
and a new larger design matrix by repeating $X$ in blocks:
\[\bbX = \ident_r\otimes X =  \left(\begin{array}{cccc} X & 0 & \dots & 0 \\ 0 & X & \dots & 0 \\ && \dots & \\ 0 & 0 & \dots & X\end{array}\right) \in\R^{nr \times pr}.\]
(Here $\otimes$ is the Kronecker product.)
Define the coefficient vector $\beta = \textnormal{vec}(B)\in\R^{pr}$ and  noise vector $\eps = \textnormal{vec}(E)\in\R^{nr}$.
Then the  multitask model~\eqnref{multitask_model} can be rewritten as
\begin{equation}\label{eqn:multitask_vector_model}y = \bbX\beta + \eps,\end{equation}
where $\eps$ follows a Gaussian model, $\eps\sim\normal(0,\bbSig)$, for
\[{\bbSig} = \Sigma\otimes\ident_n = \left(\begin{array}{ccc}\Sigma_{11}\ident_n  & \dots & \Sigma_{1r}\ident_n \\  \dots &  \dots &  \dots \\ \Sigma_{r1}\ident_n &  \dots & \Sigma_{rr}\ident_n\end{array}\right).\]
The group sparse structure of $\beta$ is determined by groups
\[G_j  = \{ j, j+p, \dots,j+ (r-1)p\}\]
for $j=1,\dots,p$; this corresponds to the row sparsity of $B$ in the original formulation~\eqnref{multitask_model}.
Then, the multitask learning problem has been reformulated into a group-sparse regression problem---and so,
the multitask lasso~\eqnref{multitask_lasso} can equivalently be solved by the group lasso optimization problem
\begin{equation}\label{eqn:multitask_as_group_lasso}\betah = \argmin_{\beta}\left\{\frac{1}{2} \norm{ y - \bbX\beta}^2_2 + \lambda\grpnorm{\beta}\right\}\;.\end{equation}

\subsection{The group false discovery rate}
The original definition of false discovery rate (FDR) is the expected proportion of incorrectly selected features among all selected features.
When the group rather than individual feature is of interest, we prefer to control the false discovery rate at the group level. Mathematically, we define the group false discovery rate ($\gfdr$) 
as
\begin{equation}
\gfdr=\EE{\frac{\#\{i:\beta_i=0,i\in \Sh\}}{\#\{i:i\in \Sh\} \vee 1}}
\end{equation}
the expected proportion of selected groups which are actually false discoveries. Here $\Sh=\{i:\betah_i\neq 0\}$ is the set of all selected group of features, while $a\vee b$ denotes $\max\{a,b\}$.

\subsection{The knockoff filter for sparse linear regression}
In the sparse (rather than group-sparse) setting, the lasso~\eqnref{lasso} provides an accurate estimate for the coefficients in a sparse linear model,
but performing inference on the results, for testing the accuracy of these estimates or the set of features selected, remains a challenging problem.
The knockoff filter \cite{barber2015} addresses this question, and provides a method controlling the false discovery rate (FDR) of the selected set
at some desired level $q$ (e.g.~$q=0.2$).

To run this method, there are two main steps: constructing knockoffs, and filtering the results.
First a set of $p$ knockoff features is constructed: for each feature $X_j$, $j=1,\dots,p$, it is given a knockoff copy $\Xp_j$, where the matrix
of knockoffs $\Xp = [\Xp_1 \ \dots \ \Xp_p]$ satisfies, for some vector $s\geq 0$,
\begin{equation}\label{eqn:knockoff_condition}
\Xp^\top \Xp = X^\top X, \ \Xp^\top X = X^\top X - \diag\{s\}.\end{equation}
Next, the lasso is run on an augmented data set with response $y$ and $2p$ many features $X_1,\dots,X_p,\Xp_1,\dots,\Xp_p$:
\[\betah(\lambda) = \argmin_{b\in\R^{2p}} \left\{\norm{y - [ X \ \Xp ] b}_2^2 +\lambda \norm{\beta}_1\right\}.\]
This is run over a range of $\lambda$ values decreasing from $+\infty$ (a fully sparse model) to $0$ (a fully dense model).
If $X_j$ is a true signal---that is, it has a  nonzero effect on the response $y$---then this should be evident
in the lasso: $X_j$ should enter the model earlier (for larger $\lambda$) than its knockoff copy $\Xp_j$. However,
if $X_j$ is null---that is, $\beta_j=0$ in the true model---then it is equally likely to enter before or after $\Xp_j$.

Next, to filter the results, let $\lambda_j$ and $\tlam_j$ be the time of entry into the lasso path for each feature and knockoff:
\[\lambda_j = \sup\{\lambda : \betah(\lambda)_j\neq 0\}, \tlam_j = \sup\{\lambda:\betah(\lambda)_{j+p}\neq 0\},\]
and let $\Sh(\lambda), \St(\lambda)\subseteq\{1,\dots,p\}$ be the sets of original features, and knockoff features, which have entered
the lasso path before time $\lambda$, and before their counterparts:
\[\Sh(\lambda) = \{j:\lambda_j > \tlam_j\vee\lambda\}\text{ and }\St(\lambda) =\{j:\tlam_j > \lambda_j\vee\lambda\}.\]
Estimate the proportion of false discoveries in $\Sh(\lambda)$ as
\begin{equation}\label{eqn:fdphat}
\fdp(\lambda)\approx\fdph(\lambda) =  \frac{|\St(\lambda)|}{|\Sh(\lambda)|\vee 1}.\end{equation}
To understand why, note that since $X_j$ and $\Xp_j$ are equally likely to enter in either order if $X_j$ is null (no real effect),
then $j$ is equally likely to fall into either $\Sh(\lambda)$ or $\St(\lambda)$. Therefore, the numerator $|\St(\lambda)|$ should 
be an (over)estimate of the number of nulls in $\Sh(\lambda)$---thus, the ratio estimates the FDP.
Alternately, we can choose a more conservative definition
\begin{equation}\label{eqn:fdphatplus}
\fdp(\lambda)\approx\fdph_+(\lambda) =  \frac{1+|\St(\lambda)|}{|\Sh(\lambda)|\vee 1}.\end{equation}

Finally, the knockoff filter selects $\hlam = \min\{\lambda:\fdph(\lambda)\leq q\}$, where $q$
is the desired bound on FDR level,  and then
outputs the set $\Sh(\hlam)$ as the set of ``discoveries''.
The knockoff+ variant does the same with $\fdph_+(\lambda)$. 
Theorems 1 and 2 of \cite{barber2015} prove that the knockoff procedure
bounds a modified form of the FDR, $\mfdr = \EE{\frac{(\text{\# of false discoveries})}{(\text{\# of discoveries}) + q^{-1}}}$, while the knockoff+
procedure bounds the FDR.

\section{The knockoff filter for group sparsity}\label{sec:knockoff_background}
In this section, we extend the knockoff method to the group sparse setting. This involves two key modifications:
the construction of the knockoffs at a group-wise level rather than for individual features, and the ``filter'' step where 
the knockoffs are used to select a set of discoveries.
Throughout the remainder of the paper, ``knockoff'' refers to the original knockoff method, while ``group knockoff'
(or, later on, ``multitask knockoff'') refers to our new method.

\subsection{Group knockoff construction}
The original knockoff construction requires that $\Xp^\top X = X^\top X - \diag\{s\}$, that is, all off-diagonal entries are equal. When the features are highly correlated, 
this construction is only possible for vectors $s$ with extremely small entries; that is, $\Xp_j$ and $X_j$ are themselves highly correlated, and the knockoff filter
then loses power as it is hard to distinguish between a real signal $X_j$ and its knockoff copy $\Xp_j$. 

In a group-sparse setting, we will see that we can relax this requirement on $\Xp^\top X$, thereby improving our power. In particular, the best gain will be in 
situations where within-group correlations are high but between-group correlations are low; this may arise in many applications, for example, when genes related
to the same biological pathways are grouped together, we expect to see the largest correlations occuring within groups rather than between genes in different groups.

To construct the group knockoffs, we require the following condition on the matrix $\Xp\in\R^{n\times p}$:
\begin{multline}\label{eqn:construct_group_knockoffs}
\Xp^\top \Xp = \Sigma\coloneqq X^\top X ,\text{ and }
\Xp^\top X = \Sigma - S,\\ \text{ where $S\succeq 0$ is group-block-diagonal,}
\end{multline}
meaning that $S_{G_i,G_j} = 0$ for any two distinct groups $i\neq j$. 
Abusing notation, write $S = \diag\{S_1,\dots,S_m\}$ where $S_i\succeq 0$ is the $p_i\times p_i$ matrix for the $i$th group block, meaning that $S_{G_i,G_i} = S_i$ for each $i$ while
$S_{G_i,G_j}=0$ for each $i\neq j$.
Extending the construction of \cite{barber2015},\footnote{This construction is for the setting $n\geq 2p$; see \cite{barber2015}
for a simple trick to extend to $n\geq p$.}
we construct these knockoffs by first selecting $S=\diag\{S_1,\dots,S_m\}$ that satisfies the condition $S\preceq 2\Sigma$, then setting
\[\Xp = X(\ident_p -\Sigma^{-1}S)+\widetilde{U}C\]
where $\widetilde{U}$ is a $n\times p$ orthonormal matrix orthogonal to the span of $X$, while $C^\top C = 2S - S \Sigma^{-1}S$ is a Cholesky decomposition.
Now, we still need to choose the matrix $S\succeq 0$, which has group-block-diagonal structure, so that the condition $S\preceq 2\Sigma$ is satisfied (this condition ensures the existence of the Cholesky decomposition defining $C$). To do this, 
we choose the following construction: we set $S = \diag\{S_1,\dots,S_m\}$ where we choose $S_i = \gamma \cdot \Sigma_{G_i,G_i}$; the scalar $\gamma\in[0,1]$ is chosen to be as large as possible so that $S\preceq 2\Sigma$ still holds, which
amounts to choosing 
\[\gamma = \min\left\{1, 2\cdot \lambda_{\min}\left(D\Sigma D\right)\right\}\]
where
$D = \diag\{\Sigma_{G_1,G_1}^{-\nicefrac{1}{2}},\dots,\Sigma_{G_m,G_m}^{-\nicefrac{1}{2}}\}$.
This construction can be viewed as an extension of the ``equivariant'' knockoff construction of~\citet{barber2015}; 
their SDP construction, which gains a slight power increase in the non-grouped setting, may
also be extended to the grouped setting but we do not explore this here.

Looking back at the group knockoff matrix condition~\eqnref{construct_group_knockoffs},
 we see that any knockoff matrix $\Xp$ satisfying~\eqnref{knockoff_condition}
would necessarily also satisfy this group-level condition. However, the group-level condition is weaker;
it allows more flexibility in constructing $\Xp$, and therefore, will enable more separation
between a feature $X_j$ and its knockoff $\Xp_j$, which in turn can increase power to detect
the true signals.

\subsection{Filter step}
After constructing the group knockoff matrix, we then select a set of discoveries (at the group
level) as follows. First, we apply the group lasso~\eqnref{glasso} to the augmented data set,
\[\betah = \argmin_{b\in\R^{2p}}\left\{\norm{y - [X \ \Xp]b}^2_2 +\lambda\grpnorm{b}\right\}.\]
Here, with the augmented design matrix $[X \ \Xp]$, we 
 we now have $2m$ many groups: one group $G_i$ for each group in the original design  matrix, 
 and one group $\widetilde{G}_i = \{j+p : j\in G_i\}$ corresponding to the same group within the 
 knockoff matrix; the penalty norm is then defined as $\grpnorm{b}=\sum_{i=1}^m \norm{b_{G_i}}_2 + \sum_{i=1}^m \norm{b_{\widetilde{G}_i}}_2$.

The filter process then proceeds exactly as for the original knockoff method,
with groups of features in place of individual features. 
First we record the time when each  group or knockoff
group enters the lasso path,
\[\lambda_i = \sup\{\lambda : \betah(\lambda)_{G_i}\neq 0\}, \tlam_i = \sup\{\lambda:\betah(\lambda)_{\widetilde{G}_i}\neq 0\},\]
then define the selected groups and knockoff groups as
\[\Sh(\lambda) = \{i:\lambda_i > \tlam_i\vee\lambda\}\text{ and }\St(\lambda) =\{i:\tlam_i > \lambda_i\vee\lambda\}\]
(note that these sets are subsets of $\{1,\dots,m\}$, the list of groups, rather than counting individual features).
Finally, 
estimate the proportion of false discoveries in $\Sh(\lambda)$ exactly as in~\eqnref{fdphat},
and define $\hlam = \min\{\lambda:\fdph(\lambda)\leq q\}$ as before; the 
final set of discovered groups is given by $\Sh(\hlam)$. (For group knockoff+, we use the more conservative
estimate of the group FDP, as for the knockoff.)

\subsection{Theoretical Results}
Here we turn to a more general framework for the group knockoff, working with the 
setup introduced in~\citet{barber2015}. Let $W\in\R^m$ be a vector of statistics,
one for each group, with large positive values for $W_i$ indicating strong evidence that
group $i$ may have a nonzero effect (i.e.~$\beta_{G_i}\neq 0$).
$W$ is defined as a function of the augmented design matrix $[X \ \Xp]$ and the response $y$,
which we write as $W = w([X \ \Xp], y)$. In the group lasso setting described above, the statistic
is given by
\[W_i = (\lambda_i \vee \tlam_i ) \cdot \sign(\lambda_i - \tlam_i).\]
In general, we require two properties for this statistic: sufficiency and group-antisymmetry. The first
is exactly as for (non-group) knockoffs; the second is a modification moving to the group sparse setting.
\begin{definition}\label{def:sufficiency} The statistic $W$ is said to obey the sufficiency property if it only depends on the Gram matrix and feature-response inner products, that is, for any $X,\Xp,y$,
\begin{equation}
w([X \ \Xp],y)=f([X,\Xp]^\top [X,\Xp], [X,\Xp]^\top y)
\end{equation}
for some function $f$.
\end{definition}

Before defining the group-antisymmetry property, we introduce some notation. For any group $i=1,\dots,m$, 
let $[ X \ \Xp]_{\textnormal{swap}(i)}$  be the matrix with
\[\left([ X \ \Xp]_{\textnormal{swap}(i)}\right)_j = \begin{cases} 
X_j,&\text{ if $1\leq j\leq p$ and $j\not\in G_i$,}\\
\Xp_j,&\text{ if $1\leq j\leq p$ and $j\in G_i$,}\end{cases}\]
and
\[\left([ X \ \Xp]_{\textnormal{swap}(i)}\right)_{j+p} = \begin{cases} 
\Xp_j,&\text{ if $1\leq j\leq p$ and $j\not\in G_i$,}\\
X_j,&\text{ if $1\leq j\leq p$ and $j\in G_i$,}\end{cases}\]
for each $j=1,\dots,p$. In other words, the columns corresponding to $G_i$ in the original component $X$,
are swapped with the same columns of $\Xp$.
\begin{definition}\label{def:group_antisymmetry}
The statistic $W$ is said to obey the group-antisymmetry property if swapping two groups $X_i$ and $\Xp_i$ has the effect of switching the sign of $W_i$ with no other change to $W$, that is,
\[w([ X \ \Xp]_{\textnormal{swap}(i)},y) = \ident^{\pm}_i \cdot w([X \ \Xp],y),\]
where $\ident^{\pm}_i$ is the diagonal matrix with a $-1$ in entry $(i,i)$ and $+1$ in all other diagonal entries.
\end{definition}

Next, to run the group knockoff or group knockoff+ method, we proceed exactly as in~\cite{barber2015};
we change notation here for better agreement with the group lasso setting. Define
\[\Sh(t) = \{i: W_i \geq t\}\text{ and }\St(t) =\{i: W_i\leq -t\}.\]
Then estimate the FDP as in~\eqnref{fdphat} for the knockoff method, or as in~\eqnref{fdphatplus} for knockoff+
(with parameter $t$ in place of the lasso penalty path parameter $\lambda$); then
find $\widehat{t}$, the minimum $t\geq 0$ with $\fdph(t)$ (or $\fdph_+(t)$) no larger than $q$, and 
output the set $\Sh = \Sh(\widehat{t})$ of discovered groups.

This procedure offers the following theoretical guarantee:
\begin{theorem}\label{thm:main_group}
If the vector of statistics $W$ satisfies the sufficiency and group-antisymmetry assumption, then 
the group knockoff procedure controls a modified group FDR,
\[\mgfdr = \EE{\frac{\#\{i:\beta_i=0,i\in \Sh\}}{\#\{i:i\in \Sh\} +q^{-1}}}\leq q,\]
while the group knockoff+ procedure controls the group FDR, $\gfdr\leq q$.
\end{theorem} 
The proof of this result follows the original knockoff proof of~\citet{barber2015},
and we do not reproduce it here; the result
is an immediate consequence of their main lemma,
moved into the grouped setting:
\begin{lemma}\label{lem:signs}
Let $\epsilon\in \{\pm 1\}^{m}$ be a sign sequence independent of $W$, with $\epsilon_i=1$ for all non-null groups $i$ and $\epsilon_i\sim \{\pm 1\}$ independently with equal probability for all null groups $i$. Then we have
\begin{equation}
(W_1,\cdots,W_m)=_{d} (W_1\epsilon_1,\cdots, W_m\epsilon_m),
\end{equation} 
where $=_{d}$ denotes equality in distribution.
\end{lemma}
This lemma can be proved via the sufficiency and group-antisymmetry properties, exactly
as for the individual-feature-level result of~\citet{barber2015}.

\section{Knockoffs for multitask learning}\label{sec:knockoff_multitask}

For the multitask learning problem, the reformulation as a group lasso problem~\eqnref{multitask_as_group_lasso}
suggests that we can apply the group-wise knockoffs to this problem as well. However, there is one immediate difficulty:
the model for the noise $\eps$ in~\eqnref{multitask_as_group_lasso} has changed---the entries of $\eps$ are not independent,
but instead follow a multivariate Gaussian model with covariance $\bbSig$.
In fact, we will see shortly that we can work even in this more general setting.
Reshaping the data to form a group lasso problem as in~\eqnref{multitask_vector_model},
we will work with the vectorized response $y\in\R^{nr}$ and the repeated-block design matrix $\bbX\in\R^{nr\times pr}$.
We will also construct a repeated-block knockoff matrix,
\[\bbXp = \ident_r\otimes \Xp=\left(\begin{array}{cccc} \Xp & 0 & \dots & 0 \\ 0 & \Xp & \dots & 0 \\ && \dots & \\ 0 & 0 & \dots & \Xp\end{array}\right),\]
where $\Xp\in\R^{n\times p}$ is any matrix satisfying the original knockoff construction conditions~\eqnref{knockoff_condition}
with respect to the original design matrix $X$.
Applying the group knockoff methodology with this data $(\bbX,y)$ and knockoff matrix $\bbXp$,
we obtain the following result:
\begin{theorem}
For the multitask learning setting with an arbitrary covariance structure $\Sigma\in\R^{r\times r}$,
the knockoff or knockoff+ methods control the modified group FDR or the group FDR, respectively,
at the level $q$.
\end{theorem}
\begin{proof}
In order to apply the result for the group-sparse setting to this multitask scenario,
we need to address two questions: first, whether $\bbXp$ satisfies the group knockoff matrix
conditions~\eqnref{construct_group_knockoffs}, and second, how to  handle
the issue of the non-\iid structure of the noise $\eps$.

We first check the conditions~\eqnref{construct_group_knockoffs} for $\bbXp$.
let $\Xp\in\R^{n\times p}$ be a knockoff matrix for $X$, satisfying~\eqnref{knockoff_condition}, and 
let $\Sigma = X^\top X$.
 Then 
we see that
\begin{multline*}\bbXp^\top\bbXp = \ident_r\otimes (\Xp^\top \Xp) = \ident_r \otimes \Sigma = \bbX^\top \bbX,\text{ and }\\
\bbXp^\top\bbX = \ident_r\otimes (\Xp^\top X) = \ident_r \otimes(\Sigma - \diag\{s\})\\ = \bbX^\top \bbX - \ident_r\otimes \diag\{s\}\end{multline*}
where $s$ is defined as in~\eqnref{knockoff_condition}. Since the difference $\ident_r\otimes\diag\{s\}$ is a diagonal matrix, we see that $\bbXp$ satisfies the group knockoff condition~\eqnref{construct_group_knockoffs}; in fact, it satisfies
the stronger (ungrouped) knockoff condition~\eqnref{knockoff_condition}.

\begin{figure*}[t]\centering
\includegraphics[width=\textwidth]{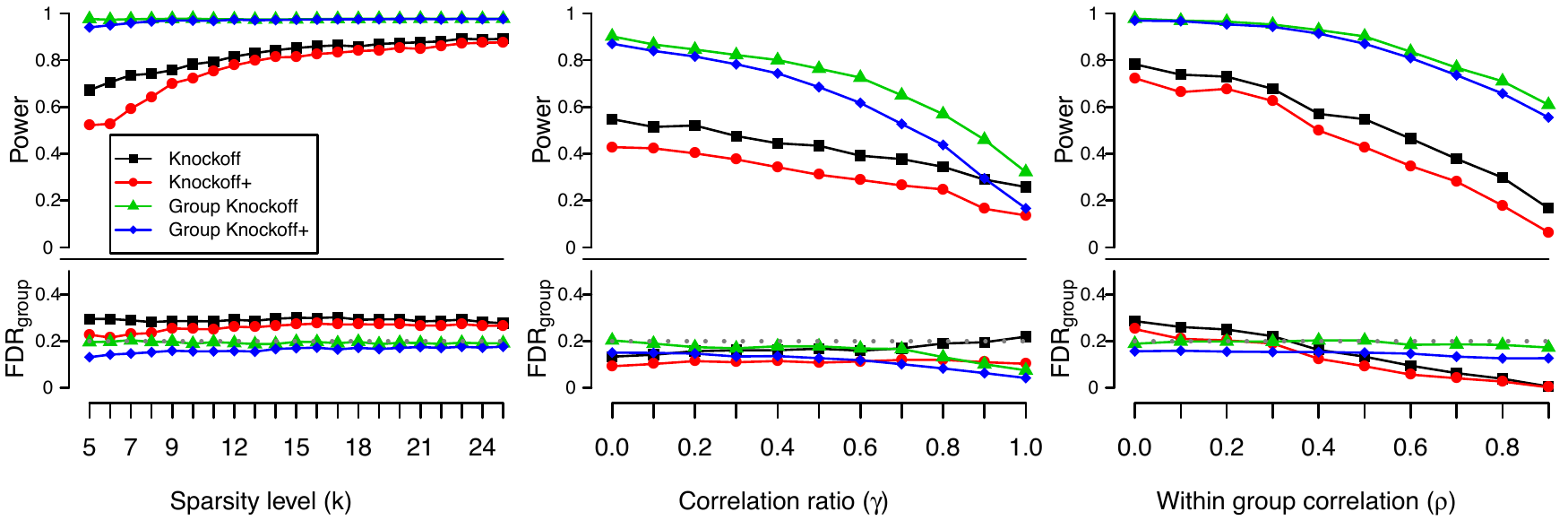}
\caption{Results for the group-sparse regression simulation, comparing group knockoff and knockoff+ against the
original knockoff and knockoff+ methods.}
\label{fig:groupsparse_sim}
\end{figure*}

Next we turn to the issue of the non-identity covariance structure $\bbSig$
for the noise term $\eps\in\R^{nr}$. First, 
write
\[\bbSig^{-\nicefrac{1}{2}} = \Sigma^{-\nicefrac{1}{2}}\otimes \ident_n\]
to denote an inverse square root for $\bbSig$.
Note also that
\begin{multline}\label{eqn:commutes}\bbSig^{-\nicefrac{1}{2}}\cdot \bbX  = (\Sigma^{-\nicefrac{1}{2}}\otimes\ident_n)\cdot (\ident_r\otimes X)
 = \Sigma^{-\nicefrac{1}{2}}\otimes X\\ =(\ident_r\otimes X)\cdot (\Sigma^{-\nicefrac{1}{2}}\otimes\ident_p) = \bbX \cdot \bbSig^{-\nicefrac{1}{2}}_*,\end{multline}
 for $\bbSig^{-\nicefrac{1}{2}}_* = \Sigma^{-\nicefrac{1}{2}}\otimes\ident_p$.
Taking our vectorized  multitask regression model~\eqnref{multitask_vector_model},  multiplying both sides by $\bbSig^{-\nicefrac{1}{2}}$ on the left,
and applying~\eqnref{commutes}, we obtain a ``whitened'' reformulation of our model,
\begin{equation}\label{eqn:multitask_whitened}
y^{\text{wh}} = \bbX \cdot (\bbSig^{-\nicefrac{1}{2}}_* \beta) + \eps^{\text{wh}}\text{ for }\begin{cases}y^{\text{wh}} = \bbSig^{-\nicefrac{1}{2}} y,\\
 \eps^{\text{wh}} =  \bbSig^{-\nicefrac{1}{2}} \eps,\end{cases}\end{equation}
where $\eps^{\text{wh}}\sim \normal(0,\ident_{nm})$ is the ``whitened'' noise.
Now we are back in a standard linear regression setting, and can apply the knockoff method---note that we are working with a new setup: while the design matrix $\bbX$ is the same as in~\eqnref{multitask_vector_model}, 
we now work with response vector $y^{\text{wh}}$ and coefficient vector $\bbSig^{-\nicefrac{1}{2}}_* \beta$.
The group sparsity of the coefficient vector has not changed, due to the block structure of $\bbSig^{-\nicefrac{1}{2}}$;
we have
\[(\bbSig^{-\nicefrac{1}{2}}_* \beta)_{G_j} = \bbSig_{G_j,G_j}^{-\nicefrac{1}{2}}\beta_{G_j}\]
for each $j=1,\dots,p$, and so the ``null groups'' for the original coefficient vector $\beta$ (i.e.~groups
$j$ with $\beta_{G_j}=0$) are preserved in this reformulated model.

We need to check only that the
group lasso output, namely $\betah$, depends on the data only through the sufficient statistics $\bbX^\top\bbX$ and $\bbX^\top y^{\text{wh}}$; 
here we use the ``whitened'' response $y^{\text{wh}}$ rather than the original response vector $y$ since the knockoff theory applies
to linear regression with \iid Gaussian noise, as in the model~\eqnref{multitask_whitened} for $y^{\text{wh}}$.
When we apply the group lasso, as in the optimization problem~\eqnref{multitask_as_group_lasso}, it is clear
that the minimizer $\betah$ depends on the data $\bbX,y$ only through $\bbX^\top\bbX$ and $\bbX^\top y$. Furthermore, we can write
\[\bbX^\top y = \bbX^\top \bbSig^{\nicefrac{1}{2}} y^{\text{wh}} = \bbSig^{\nicefrac{1}{2}}_*\cdot (\bbX^\top y^{\text{wh}}),\]
where we can show $\bbSig^{\nicefrac{1}{2}}\cdot \bbX = \bbX\cdot \bbSig^{\nicefrac{1}{2}}_*$ exactly as in~\eqnref{commutes} before. Therefore,
$\betah$, depends on the data only through the sufficient statistics $\bbX^\top\bbX$ and $\bbX^\top y^{\text{wh}}$, as desired.

Our statistics for the knockoff filter therefore will  satisfy the sufficiency property. 
The group-antisymmetry property
is obvious from the definition of the method.
Therefore, applying our main result \thmref{main_group} for the group-sparse setting
to the whitened model~\eqnref{multitask_whitened}, we see that the (modified or unmodified) group FDR
control result holds for this setting.
\end{proof}

\section{Simulated data experiments}
We test our methods in the group sparse and multitask settings. All experiments were carried
out in Matlab \cite{MATLAB} and R \cite{R}, including the \texttt{grpreg} package in R \cite{breheny2015group}.

\subsection{Group sparse setting}
To evaluate the performance of our method in the group sparse setting,
we compare it empirically with the (non-group) knockoff using simulated data from a group sparse linear regression,
and examine the effects of sparsity level and feature correlations within and between groups.

\subsubsection{Data}
To generate the simulation data, we use the sample size $n=3000$ with number of features $p=1000$. In our basic setting, 
the number of groups is $m=200$ with corresponding number of features per group set as $p_i=5$ for each group $i$. To generate features, as a default
we use an uncorrelated setting, drawing the entries of $X$ as \iid~standard normals, then normalize the columns of $X$. 
Our default sparsity level is $k=20$ (that is, $k$ groups with nonzero signal); $\beta_j$, for each $j$ inside a signal group, id chosen randomly from $\{\pm 3.5\}$.

To study the effects of sparsity level and feature correlation, we then vary these default settings as follows (in each experiment,
one setting is varied while the others remain at their default level):
\begin{itemize}
\item Sparsity level: we vary the number of groups with nonzero effects, $k\in\{10,12,14,\dots,50\}$.
\item Between-group correlation: we fix within-group correlation $\rho = 0.5$, and set the between-group correlation to be
$\gamma\rho$, with $\gamma\in\{0,0.1,0.2,\dots,0.9\}$. We then
draw the rows of $X\in\R^{n\times p}$ independently from a multivariate normal distribution with mean $0$ and covariance matrix $\Sigma$, with diagonal entries $\Sigma_{jj}=1$,
within-group correlations $\Sigma_{jk}=\rho$ for $j\neq k$ in the same group, and between-group correlations
$\Sigma_{jk}=\gamma\rho$
  for $j,k$ in different groups. 
   Afterwards, we normalize the columns of $X$.
\item Within-group correlation: as above, but we fix $\gamma=0$ (so that between-group correlation is always zero)
and vary within-group correlation, with $\rho\in\{0,0.1,\dots,0.9\}$.
\end{itemize}

For each setting, we use target FDR level $q=0.2$ and repeat each experiment $100$ times.

\begin{figure*}[t]\centering
\includegraphics[width=0.24\textwidth]{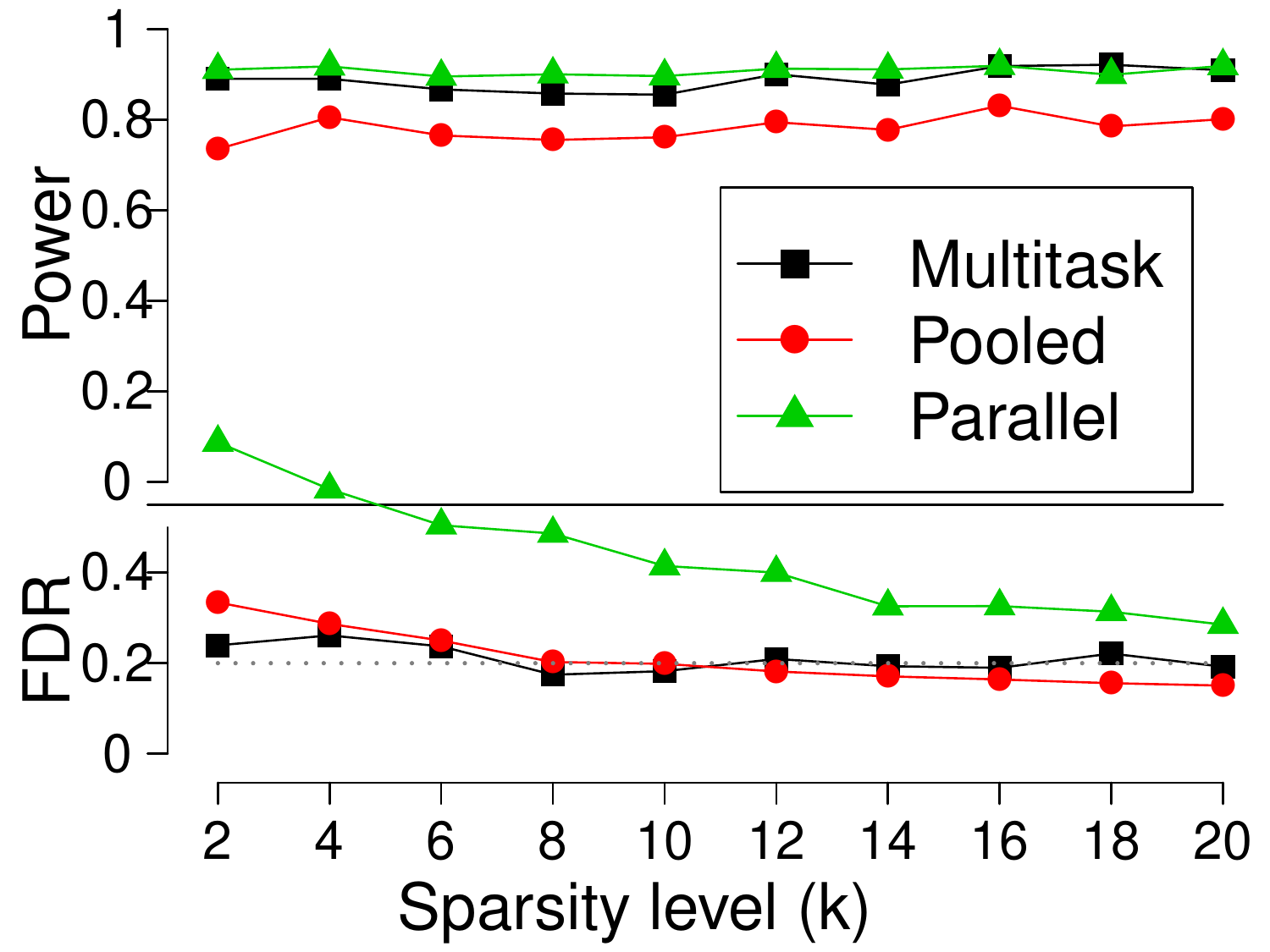}
\includegraphics[width=0.24\textwidth]{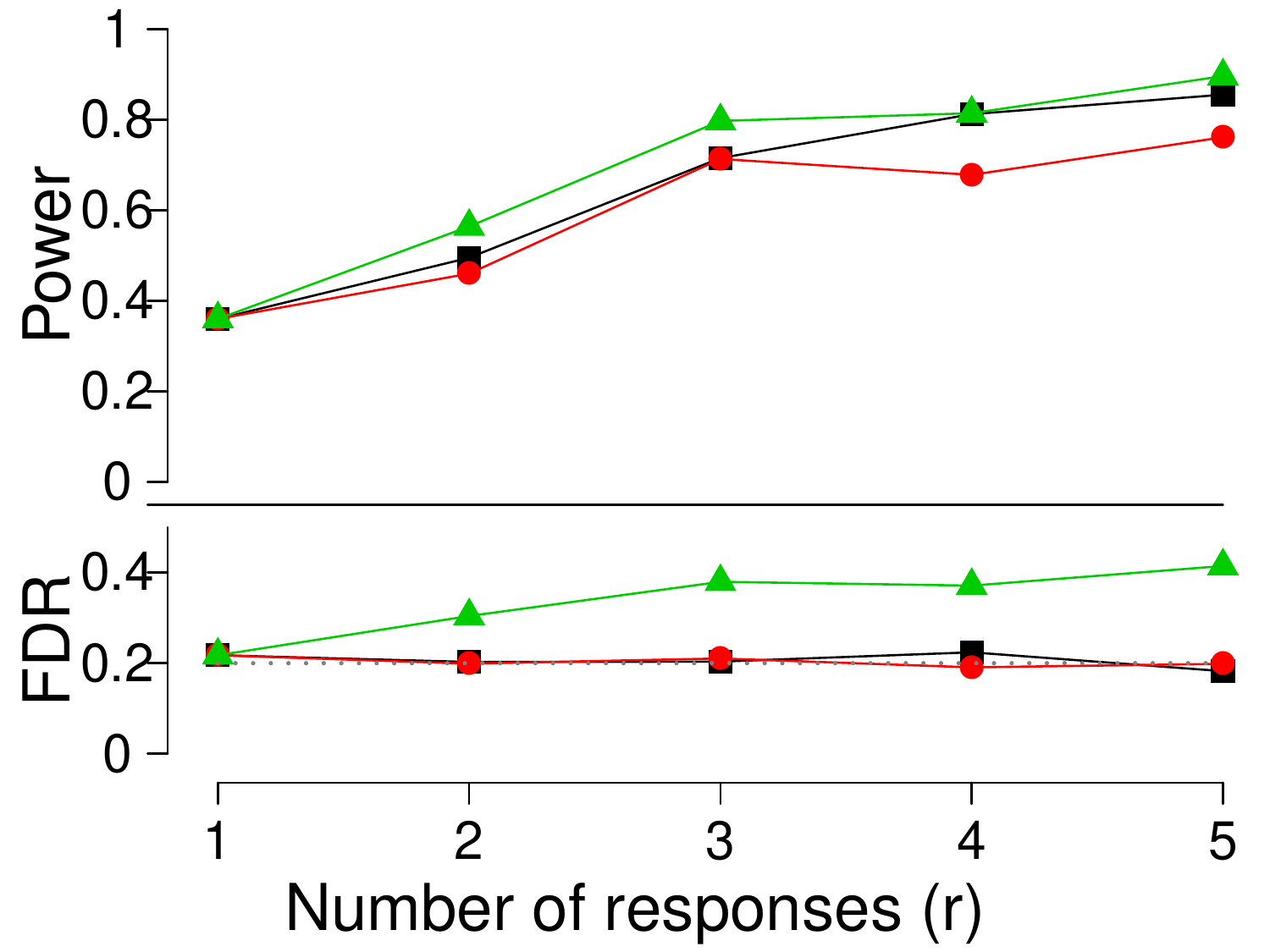}
\includegraphics[width=0.24\textwidth]{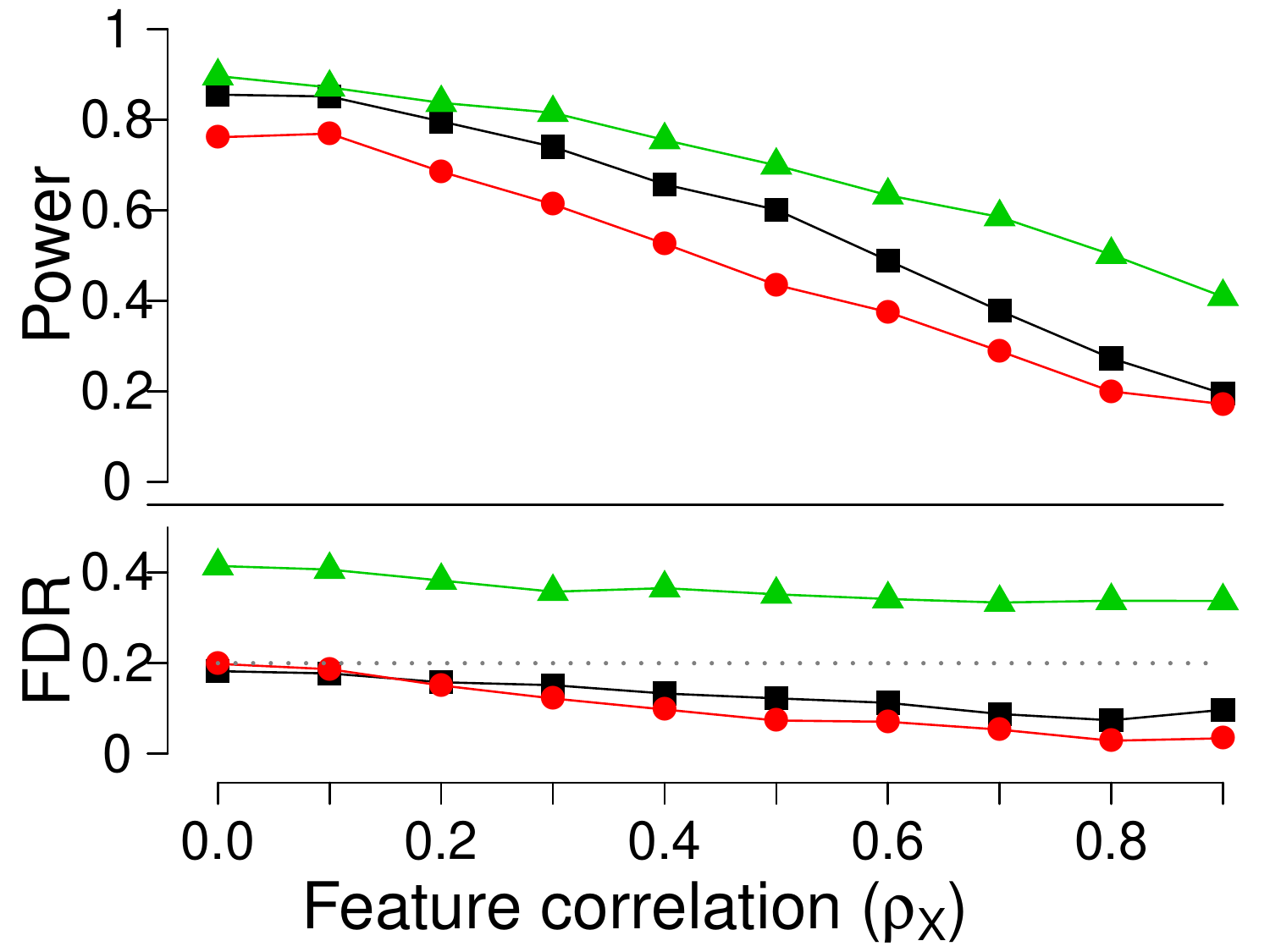}
\includegraphics[width=0.24\textwidth]{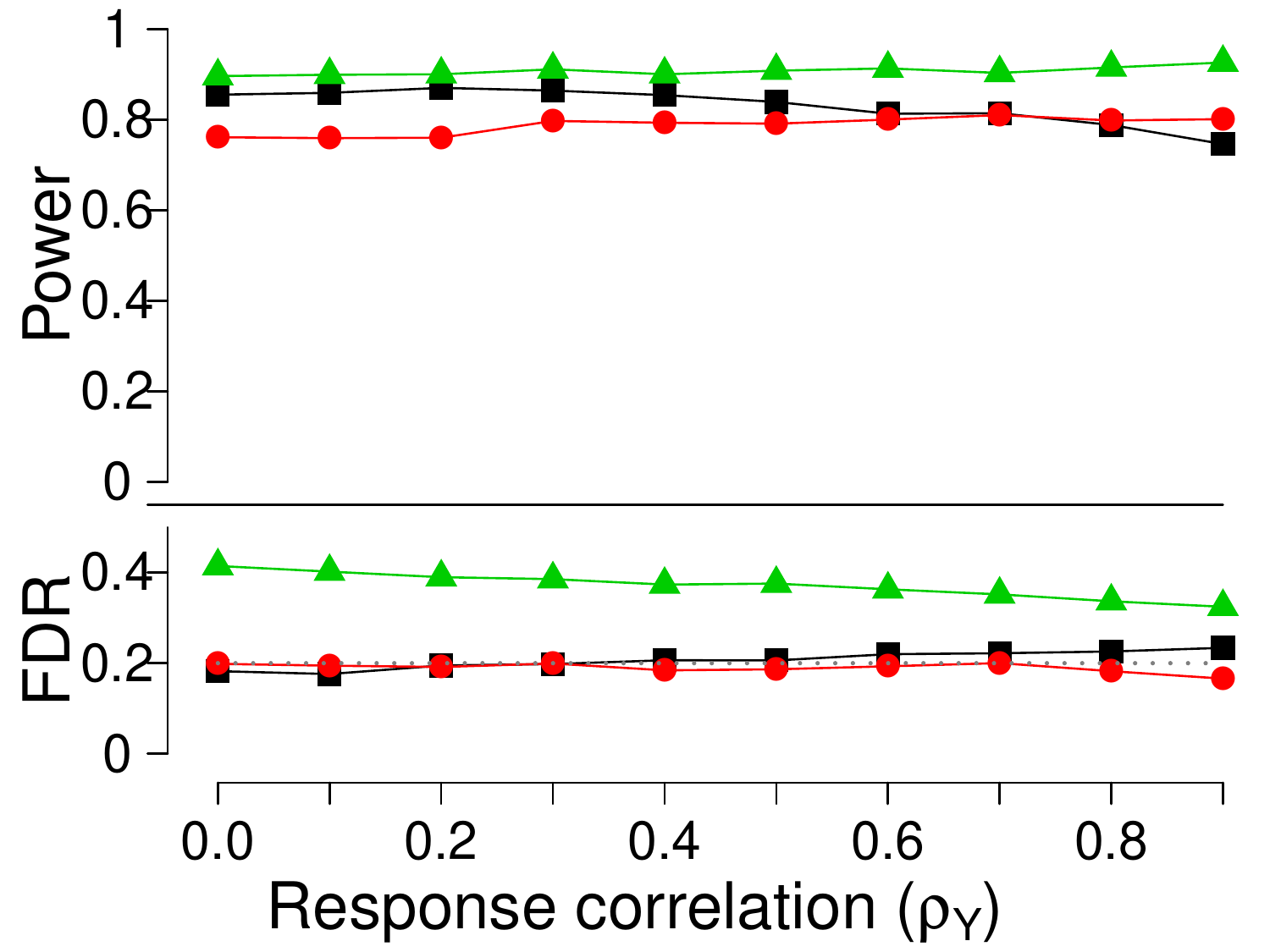}
\caption{Results for the multitask regression simulation, comparing multitask knockoff  with the
pooled and parallel knockoff methods.}
\label{fig:multitask_sim}
\end{figure*}

\subsubsection{Results}

Our results are displayed in \figref{groupsparse_sim}, which displays power (the proportion of true signals
which were discovered) and FDR
at the group level, averaged over all trials. We see that all four methods successfully control FDR at the desired
level. Across all settings, the group knockoff is more powerful than the knockoff, showing the benefit of 
leveraging the group structure. The group knockoff+ and knockoff+ are each slightly more conservative
than their respective methods without the ``+'' correction. 
From the experiments with zero between-group correlation and increasing within-group correlation $\rho$,
 we see that  knockoff has rapidly decreasing power as $\rho$ increases, while group knockoff does not show much power loss.
This highlights the benefit of the group-wise
construction of the knockoff matrix; for the original knockoff,
high within-group correlation forces the knockoff features $\Xp_j$ to be nearly equal to the $X_j$'s, but
this is not the case for the group knockoff construction and the greater separation allows high power to be maintained.

\subsection{Multitask regression setting}
To evaluate the performance of our method in the multitask regression setting, we next
perform a simulation to compare the multitask knockoff
with the  knockoff. (For clarity in the figures, we do not present results for the knockoff+ versions of these methods;
the outcome is predictable, with knockoff+ giving slightly better FDR control but lower power.)
For the multitask knockoff, we implement the method exactly as described in \secref{knockoff_multitask}.
The $j$th feature is considered a discovery if the corresponding group is selected.
For the  knockoff, we use the group lasso formulation of the multitask model, given in~\eqnref{multitask_vector_model},
and apply the knockoff method to the reshaped data set $(\bbX,y)$; we call this the ``pooled''  knockoff.
We also run the  knockoff separately on each of the $r$ responses (that is,
we run the knockoff with data $(X,Y_j)$ where $Y_j$ is the $j$th column of $Y$, separately for $j=1,\dots,r$).
We then combine the results: the $j$th feature is considered a discovery if it is selected in any of 
the $r$ individual regressions; this version is the ``parallel''  knockoff.

\subsubsection{Data}
To generate the data, our default settings for the multitask model given in~\eqnref{multitask_model}
are as follows: we set 
the sample size $n=150$, the number of features $p=50$, with $m=5$ responses. The true matrix 
of coefficients $B$ has its $k=10$ rows nonzero, which are chosen as $2\sqrt{m}$ times a random 
unit vector. 
The design matrix $X$ is generated by drawing \iid~standard normal entries and then normalizing
the columns,
and the entries of the error matrix $E$ are also \iid standard normal.
We set the target FDR level at $q=0.2$ and repeat all experiments $100$ times. These default
settings will then be varied in our experiments to examine the roles of the various parameters (only
one parameter is varied at a time, with all other settings at their defaults):
\begin{itemize}
\item Sparsity level: the number of nonzero rows of $B$ is varied, with $k\in\{2,4,6,\dots,20\}$.
\item Number of responses: the number of responses $r$ is varied, with $r\in\{1,2,3,4,5\}$.
\item Feature correlation: 
the rows of $X$ are \iid~draws
from a $N(0,\Sigma_X)$ distribution, with a tapered covariance matrix which has entries $(\Sigma_X)_{jk} = (\rho_X)^{|j-k|}$,
with $\rho_X\in\{0,0.1,0.2,\dots,0.9\}$. (The columns of $X$ are then normalized.)
\item Response correlation: 
the rows of the noise $E$ are \iid~draws
from a $N(0,\Sigma_Y)$ distribution, with a equivariant correlation structure which has entries $(\Sigma_Y)_{jj}=1$ for all $j$,
and $(\Sigma_Y)_{jk}=\rho_Y$ for all $j\neq k$,
with $\rho_Y\in\{0,0.1,0.2,\dots,0.9\}$.
\end{itemize}

\subsubsection{Results}
Our results are displayed in \figref{multitask_sim}. For each method, we display the resulting FDR and power
for selecting features with true effects in the model.
The parallel  knockoff is not able to control the FDR. This may be due to the fact that 
this method combines discoveries across multiple responses; if the true positives selected for each response tend to overlap,
while the false positives tend to be different (as they are more random), then the false discovery
proportion in the combined results may be high even though it should be low for each individual responses' selections.
Therefore, while it is more powerful than the other methods, it does not lead to reliable
FDR control. 
Turning to the other methods, both multitask knockoff and pooled  knockoff generally control
FDR at or near $q=0.2$ except in the most challenging (lowest power) settings, where as expected
from the theory, the FDR exceeds its target level. 
Across all settings, multitask knockoff is more powerful than pooled  knockoff, and same for the two variants of knockoff+.
Overall we see the advantage in the multitask formulation, with which we are able to
identify a larger number of discoveries while maintaining FDR control.

\section{Real data experiment}

We next apply the knockoff for multitask regression to a real data problem. We study a
data set that seeks to idenitify drug resistant mutations in HIV-1 \cite{rhee2006genotypic}.
 This data set was analyzed by \cite{barber2015}
using the  knockoff method. Each observation, sampled from a single individual, identifies mutations along
various positions in the protease or reverse transcriptase (two key proteins) of the virus,
 and measures  resistance against a range of different
drugs from three classes: protease inhibitors (PIs), nucleoside reverse transcriptase inhibitors (NRTIs),
and nonnucleoside reverse transcriptase inhibitors (NNRTIs). 
In \cite{barber2015} the data for each drug was analyzed 
separately; the response $y$ was the resistance level to the drug while the features $X_j$ were markers for the presence
or absence of the $j$th mutation. Here, we apply the multitask knockoff to this problem: for each class of drugs,
since the drugs within the class have related biological mechanisms, we expect the sparsity pattern (i.e.~which mutations
confer resistance to that drug) to be similar across each class. We therefore have a matrix of responses, $Y\in\R^{n\times r}$,
where $n$ is the number of individuals and $r$ is the number of drugs for that class. We compare our results to
those obtained with the  knockoff method where drugs are analyzed one at a time (the ``parallel''
 knockoff from the multitask simulation).

\subsection{Data}
Data is analyzed separately for each of the three drug types.
To combine the data across different drugs, we first remove any drug with a high
proportion of missing drug resistance measurements; this results in two PI drugs and 
one NRTI drug being removed (each with over 35\% missing data). The remaining drugs all have $<10\%$
missing data; many drugs have only $1-2\%$ missing data. Next we remove data from any
individual that is missing drug resistance information from any of the (remaining) drugs. 
Finally, we keep only those mutations which appear $\geq 3$ times in the sample. 
The resulting data set sizes  are:
\begin{table}[ht]\small
\centering
\begin{tabular}{cccc}
\hline
Class &\#  drugs ($r$) &\#  observations ($n$)& \#  mutations ($p$) \\
\hline
PI &5 &701 &198\\
NRTI &5 &614 &283\\
NNRTI &3 &721 &308\\
\hline
\end{tabular}
\end{table}

\subsection{Methods}
For each of the three drug types, we form the $n\times r$ response matrix $Y$
by taking the log-transformed drug resistance measurement for the $n$ individuals and the $r$
drugs, and the $n\times p$ feature matrix $X$ recording which of the $p$ mutations are present
in each of the $n$ individuals. We then apply the multitask knockoff as described
in \secref{knockoff_multitask}, with target FDR level $q = 0.2$. For comparison,
we also apply the  knockoff to the same data (analyzing each drug separately),
again with $q=0.3$. We use the equivariant construction for the knockoff matrix for 
both methods. 

\begin{figure}[t]\centering
\includegraphics[width=\textwidth]{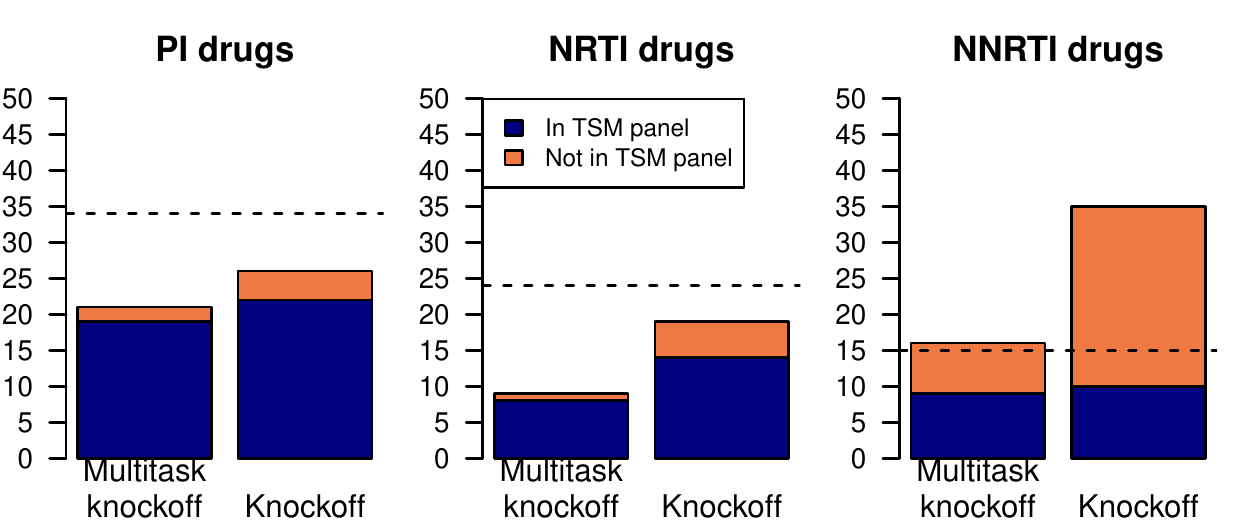}\vspace{-.2in}
\caption{Results on the HIV-1 drug resistance data set. For each drug class, we plot the number of protease positions (for PI) or
reverse transcriptase (RT) positions (for NRTI or NNRTI)  which were selected by the multitask knockoff or 
knockoff method. The color indicates whether or not the selected position appears in the treatment
selected mutation (TSM) panel, and the horizontal line shows the total number of positions
on the TSM panel.}
\label{fig:hiv_results}
\end{figure}

\subsection{Results}
We report our results by comparing the discovered mutations, within each drug class, against
the treatment-selected mutation (TSM) panel \cite{rhee2005hiv}, which gives mutations associated
with treatment by a drug from that class. 
As in \cite{barber2015} we report the counts by position rather than by mutation, i.e.~combinining
all mutations discovered at a single position, since multiple mutations at 
the same position are likely to have related effects.
To compare with the  knockoff method, for each drug class we consider mutation $j$ to be 
a discovery for that drug class, if it was selected for any of the drugs in that class.
The results are displayed in \figref{hiv_results}. In this experiment, we see that the multitask knockoff
has somewhat fewer discoveries than the  knockoff, but seems to show better
agreement with the TSM panel. As in the multitask simulation, this may be due to the fact that the  knockoff
combines discoveries across several drugs; a low false discovery proportion for each drug
individually can still lead to a high false discovery proportion once the results are combined.

\section{Discussion}
We have presented a knockoff filter for the group sparse regression and multitask
regression problems, where sharing information within each group or across
the set of response variables allows for a more powerful feature selection method.
Extending the knockoff framework to other structured estimation problems,
such as non-linear regression or to low-dimensional latent structure other than sparsity,
would be interesting directions for future work.

\bibliographystyle{imsart-nameyear}
\bibliography{bib}


\end{document}